\documentclass[12pt]{amsart}

\usepackage[utf8]{inputenc}
\usepackage[english]{babel}

\usepackage{amsmath}
\usepackage{amsfonts}
\usepackage{amssymb}
\usepackage{amsthm}

\author{M.\,Ziatdinov}
\title{Quantum hashing. Group approach}

\newcommand{\ZZ}{\mathbb{Z}}
\newcommand{\KK}{\mathbb{K}}

\newcommand{\ket}[1]{|#1\rangle}

\newcommand{\braket}[2]{\langle#1|#2\rangle}
\newcommand{\braaket}[3]{\langle#1|#2|#3\rangle}
\newcommand{\brz}[1]{\langle \psi_0 | #1 | \psi_0 \rangle}

\newcommand{\KG}{K_\mathrm{good}}
\newcommand{\Aut}{\mathrm{Aut}}

\newcommand{\HH}[1]{(\mathcal{H}^2)^{\otimes #1}}
\newcommand{\HtoH}[1]{[\HH{#1} \to \HH{#1}]}
\newcommand{\Mean}[2]{\mathbf{E}_{#1}\bigg[#2\bigg]}
\newcommand{\Prob}{\mathrm{P}}

\newtheorem{thm}{Theorem}
\newtheorem{lemma}{Lemma}
\newtheorem{defn}{Definition}

\begin{document}

\maketitle

\begin{abstract}
In this paper we consider a generalization of quantum hash
  functions for arbitrary groups. We show that quantum hash function
  exists for arbitrary abelian group. We construct a set of ``good''
  automorphisms --- a key component of quantum hash funciton. We prove
  some restrictions on Hilbert space dimension and group used in
  quantum hash function
\end{abstract}

\section{Introduction}
\label{sec:intro}

Buhrman et al.\ in~\cite{BCWdW01} introduced the notion of quantum fingerprinting and constructed first quantum hash
function. Ablayev and Vasiliev in~\cite{AV09} offered another version of quantum fingerprinting.

In~\cite{AV14} construction of Buhrman et al.\ and Ablayev-Vasiliev's construction are generalized. It is shown that
both approaches can be viewed as composition of ``quantum generator'' and (classical) universal hash function. Also the
notion of ``quantum hash function'' is introduced.

We present an algebraic generalization of Ablayev-Vasiliev's construction. Main reason of it is maximal abstraction
while retaining such properties of quantum hash function as simple evaluation, ability to continue computing hash (i.e.
hash value of string concatenation can be somehow evaluated based on hash of first string and second string), simple
reverse transform (in Ablayev-Vasiliev's construction it is enough to reverse input string and change the sign of
rotations). 


In section~\ref{sec:def} required definitions are introduced, sections~\ref{sec:zq}-\ref{sec:abelian-group} are devoted
to contruction of quantum hash function for arbitrary abelian group, section~\ref{sec:good-automorphisms} proves
existence of ``good'' automorphisms which are key component of quantum hash function, section~\ref{sec:restrictions} is
devoted to some restrictions on possible combinations of parameters of quantum hash function.

\section{Definition}
\label{sec:def}

We start with recalling basic definitions that we will need in the paper

Let $x$ be a $n$-bit message: $x \in \{0,1\}^n$. 

Let us consider functions mapping $\{0,1\}^n$ to some (arbitrary finite) group $G$ with group operation $\circ$ and unit
element $e$:
\[
h: \{0,1\}^n \to G
\]

\smallskip

Let us choose a homomorphism \( f: G \to \HtoH{m} \), i.e. function that preserves group structure:
\[
f(g_1 \circ g_2) = f(g_1) f(g_2).
\]

We use $\HtoH{m}$ notation for a set of all unitary transformations on $m$ qubits.

\smallskip

Let us choose a set of automorphisms $\KK$ from group of all automorphisms $\Aut(G)$: 
\begin{equation}\label{eqn:auto-set}
k_i \in \KK \subseteq \Aut(G), \qquad 1 \le i \le T.
\end{equation}

We will use notation $k\{g\}$ for image of $g$ under automorphism $k$.

\medskip

We generalize notion introduced in \cite{AV09} as follows. We call set $\KG$ of elements of chosen $\KK$ ``good'' set if
for each non-unit group element $g$: 
\begin{equation}\label{good-set}
\forall g \in G, g \neq e : \frac{1}{|\KG|^2} \left| \sum_{k \in \KG} \brz{f(k\{g\})} \right|^2 < \epsilon
\end{equation}

Let us also require that for each group element:
\begin{equation}\label{eqn:mean-auto}
\forall g \in G : \Mean{k \in \KK}{\brz{f(k\{g\})}} = 0
\end{equation}

In section \ref{sec:good-automorphisms} we will show that this requirement involves existence of ``good'' set $\KG$ of
elements of $\KK$. Also we will show that this set can be chosen randomly with high
probability 

In the remaining part of this section we will consider that ``good'' subset $\KG = \{k_1,\ldots,k_t\}$ is constructed
and its elements are some automorphisms $k_j \in \KK, 1 \le j \le t$.

\bigskip


Let us define quantum hash function as follows.
\begin{defn}\label{defn:hash-function}
  Quantum hash function based on classical hash function $h$ mapping
  $X^n$ to group $G$, ``good'' set of automorphisms $K = \{k_0,
  \ldots, k_{t-1}\}$ 
  and homomorphism $f$ to space $\HtoH{m}$:
\[
\ket{\Psi_{h,G,K,f,m,\ket{\Psi_0}}(x)} = \frac{1}{\sqrt t} \sum_{j=0}^{t-1} \bigg( \ket{j} \otimes f\big( k_j\{h(x)\} \big) \ket{\psi_0} \bigg).
\]
\end{defn}

We need homomorphism and automorphisms here, because we want to preserve group structure. It will allow us easily
compute hash of string concatenation and invert quantum hash function: e.g. to compute hash of string concatenation one
need to compute hash of first string and feed it as initial state to second string’s hash computation:

\[
\ket{\Psi_{h,G,K,f,m,\ket{\Psi_0}}(u \cdot v)} = \ket{\Psi_{h_1,G,K,f,m,\ket{\Psi_1}}(v)},
\]
where $h$ can be represented as $h(u \cdot v) = h(u) \circ h(v)$, and
\[
\ket{\Psi_1} = \ket{\Psi_{h,G,K,f,m,\ket{\Psi_0}}(u)}.
\]

To reverse quantum hash function one need to negate $h(x)$.

In the rest of the article we will omit $\ket{\Psi_0}$ parameter if its value is clear from the context.

\medskip

Let us consider square of scalar product of quantum hash function
values on different inputs:
\[
\left| \braket{\Psi_{h,G,K,f,m}(x)}{\Psi_{h,G,K,f,m}(x')} \right|^2 =
\]
\[
= \left| \frac{1}{t} \sum_{j=0}^{t-1} \bigg( \braket j j 
\braaket{\psi_0}{f^\dagger(k_j\{h(x)\})f(k_j\{h(x')\})}{\psi_0} \bigg) \right|^2 = 
\]
\[
 = \left| \frac{1}{t} \sum_{j=0}^{t-1} \brz{f(k_j\{h^{-1}(x) \circ h(x')\})} \right|^2
\]

If hash function $h$ has no collision, and $h(x) \neq h(x')$, product $h(x') \circ
h(x)^{-1}$ will be equal to some element $g \neq e$ of group $G$, and
by definition of ``good'' subset $\KG$ square of scalar product will
be equal to:
\[
\left| \frac{1}{t} \sum_{j=0}^{t-1} \brz{f(k_j\{h^{-1}(x) \circ h(x')\})} \right|^2 < \epsilon.
\]
Otherwise, square of scalar product equals to one:
\[
\left| \frac{1}{t} \sum_{j=0}^{t-1} \brz{f(k_j\{h^{-1}(x) \circ h(x')\})} \right|^2
= \left| \frac{1}{t}\sum_{j=0}^{t-1} \brz{f(k_j\{e\})} \right|^2 =
\]
\[
= \left| \frac{1}{t}\sum_{j=0}^{t-1} \brz{f(e)} \right|^2 = \brz{f(e)} = \braket{\psi_0}{\psi_0}
= 1
\]



\section{Example: $\ZZ_q$}
\label{sec:zq}

Fingerprinting technique suggested in \cite{AV09} can be considered as
special case of described scheme. In other words,
\begin{lemma}\label{thm:zz-q-hash}
  There exists quantum hash function for $\ZZ_q$ group, some set of
  automorphisms $\KG$ and homomorphism into $\HH{m}$ space.
\end{lemma}

We will use $\ZZ_q$ as group $G$ and $x \mod q$ as hash function $h(x)$.

Required homomorphism $f(g)$ of group $G$ into space $\HH{1}$ is qubit
rotation around Y axis on $\frac{2\pi g}{q}$ angle. It is
homomorphism, because product of two rotations on angle $\frac{2\pi
  g}{q}$ and on angle $\frac{2\pi g'}{q}$ is rotation on angle
$\frac{2\pi (g+g')}{q}$, where sum is modulo $q$, i.e. in $\ZZ_q$ group.

The group of automorphisms of $\ZZ_q$ group is $\ZZ_q^\times$ group:
\[
\Aut(\ZZ_q) = \ZZ_q^\times.
\] 
So, we choose $\KK \subseteq \Aut(G)$ to be a set of multiplications to $\ZZ_q^\times$ 
elements. It is easy to show that condition (\ref{eqn:mean-auto}) holds:
\[
\forall g \in \ZZ_q : \Mean{k \in \ZZ_q^\times}{\Big| \exp \big\{ \frac{2\pi k g}{q} \big\} 
\ket{0} \Big|} = 0
\]

\section{Example: $G_1 \times G_2$ group}
\label{sec:group-prod}

Elements of group $G_1 \times G_2$ are pairs of corresponding group elements. Group operation 
is defined component-wise:
\[
(g_1, g_2) \circ (h_1, h_2) = (g_1 \circ_1 h_1, g_2 \circ_2 h_2)
\]

Unit element is pair $(e_1, e_2)$ of corresponding group units.

\bigskip

Let $f_1$ and $f_2$ be homomorphisms from $G_1$ to $\HH{m_1}$ and from $G_2$ to $\HH{m_2}$,
correspondingly. Let $\KK_1 \subseteq \Aut(G_1)$ and $\KK_2 \subseteq \Aut(G_2)$ be corresponding
automorphisms that satisfy condition \ref{eqn:mean-auto}.

Let us define homomorphism $f$ from $G = G_1 \times G_2$ to $\HH{(m_1+m_2)}$ as follows:
\[
f((g_1, g_2)) = f_1(g_1) \otimes f_2(g_2).
\]

Let us choose automorphism set $\KK \subseteq \Aut(G)$ as follows:
\[
\KK = \left\{ (g_1, g_2) \mapsto (k_1\{g_1\}, k_2\{g_2\}) : k_1 \in \KK_1, k_2 \in \KK_2 \right\}.
\]

Let us show that condition (\ref{eqn:mean-auto}) is satisfied for this set, if it is satisfied for
$\KK_1$ and $\KK_2$:
\begin{gather*}
\Mean{k \in \KK}{|f(k\{g\})\ket{0}|} = \Mean{(k_1,k_2) \in \KK_1\times\KK_2}{f_1(k_1\{g_1\})\ket{0} \otimes
  f_2(k_2\{g_2\})\ket{0}} = \\
= \frac{1}{|\KK_1|} \sum_{k_1 \in \KK_1} \Bigg( \frac{1}{|\KK_2|} \sum_{k_2 \in \KK_2} f_1(k_1\{g_1\})\ket{0} \otimes
f_2(k_2\{g_2\})\ket{0} \Bigg) = \\
= \frac{1}{|\KK_1|} \sum_{k_1 \in \KK_1} f_1(k_1\{g_1\})\ket{0} \cdot \frac{1}{|\KK_2|} \sum_{k_2 \in \KK_2}
f_2(k_2\{g_2\})\ket{0} = \\
= \Mean{k_1 \in \KK_1}{f_1(k_1\{g_1\})\ket{0}} \cdot \Mean{k_2 \in \KK_2}{f_2(k_2\{g_2\})\ket{0}} = 0
\end{gather*}

Thus, holds
\begin{lemma}\label{thm:group-mult-hash}
  If there exists quantum hash functions for $G_1$, $G_2$ groups in $\HH{m_1}$ and $\HH{m_2}$ spaces
  with $f_1$ and $f_2$ homomorphisms and $\KK_1$ and $\KK_2$ automorphism sets satisfying
  \ref{eqn:mean-auto}, correspondingly, we can define quantum hash function for $G_1 \times G_2$ in
  $\HH{(m_1+m_2)}$ space.
\end{lemma}

\section{Example: arbitrary abelian groups}
\label{sec:abelian-group}

\begin{thm}\label{thm:abelian}
For arbitrary abelian group $G$ there exists quantum hash function with some automorphism set $\KG$
and homomorphism $f$.
\end{thm}

We can decompose every abelian group $G$ as follows:
\[
G = \ZZ_{p_1^{\sigma_1}} \otimes \cdots \otimes \ZZ_{p_t^{\sigma_t}},
\]
thus we can apply lemma \ref{thm:zz-q-hash} to define quantum hash function in each of $\ZZ_{p_1^{\sigma_1}}$, \ldots,
$\ZZ_{p_t^{\sigma_t}}$ groups, and then apply lemma \ref{thm:group-mult-hash} to compose these hash functions to hash
function for $G$.


\section{``Good'' automorphism subsets}\label{sec:good-automorphisms}

In this section we will consider that homomorphism $f: G \to \HtoH{m}$ and automorphism set $\KK$
are chosen:
\[
k_i \in \KK \subseteq \Aut(G), \qquad 1 \le i \le T,
\]
such that for each group element holds:
\[
\forall g \in G : \Mean{k \in \KK}{\brz{f(k\{g\})}} = 0
\]

Then holds
\begin{thm}\label{thm:good-auto}
  For random set $K$ of elements of set $\KK$ (each element is chosen uniformly at random from $\KK$) the probability of
  $K$ being ``bad'' for some $g$ does not exceed $1/|G|$:
\[
\Prob\Bigg( \frac{1}{|\KG|^2} \left| \sum_{k \in \KG} \brz{f(k\{g\})} \right|^2 \ge \epsilon \Bigg) \le \frac{1}{|G|}
\]
\end{thm}

\begin{proof}
Let us consider random variables
\[
X_i = \brz{f(k_i\{g\})},
\]
where $\{ k_1, k_2, \ldots, k_t \} = K$, and
\[
Y_t = \sum_{i=1}^t X_i.
\]

Let us show that sequence $Y_0 = 0, Y_1, Y_2, \ldots, Y_{|K|}$ is a martingale (cf.~\cite{McDiarmid89}).

Let us show that expectation value of $Y_t$ exists and is finite.
\begin{equation}\label{eqn:mean-yt-finite}
\Mean{}{Y_t} = \sum_{i=1}^t \Mean{}{X_i} = \sum_{i=1}^t \Mean{}{\brz{f(k_i\{g\})}} = 0
\end{equation}

Let us show that conditional expectation of $Y_t$ given $Y_{t-1}, \ldots, Y_1, Y_0$ is equal to $Y_{t-1}$:
\begin{equation}\label{eqn:cond-mean-yt}
\Mean{}{Y_t | Y_{t-1}, \ldots, Y_0} = \frac{1}{|K|} \sum_{i=1}^{|K|} \big( Y_{t-1} +
\brz{f(k_i\{g\})} \big) = Y_{t-1} + \Mean{}{X_i} = Y_{t-1}
\end{equation}

Because (\ref{eqn:mean-yt-finite}) and (\ref{eqn:cond-mean-yt}) hold, $Y_t$ is a martingale.

Let us estimate the difference of $Y_t$ and $Y_{t-1}$:
\[
|Y_t - Y_{t-1}| = \Big| \sum_{i=1}^t X_i - \sum_{i=1}^{t-1} X_i \Big| = | X_t | \le 1 
\]

Let us apply Azuma inequality:
\begin{equation}\label{eqn:yt-azuma}
\Prob\big( |Y_{|K|} - Y_0| > \lambda \big) = \Prob\Big( \big| \sum_{i=1}^{|K|} X_i \big| > \lambda
\Big) 
\le 2 \exp \left\{ - \frac{\lambda^2}{2|K|} \right\}
\end{equation}

Let
\[
\lambda = \sqrt{\epsilon} |K|.
\]

Then Azuma inequality (\ref{eqn:yt-azuma}) will take the form:
\[
\Prob\Big( \big| \sum_{i=1}^{|K|} X_i \big| > \sqrt{\epsilon} |K| \Big) \le 
2 \exp \left\{ - \frac{\epsilon |K|^2}{2} \right\} \le \frac{1}{|G|}
\]
where $|K| \ge \frac{2}{\epsilon} \ln |G|$.

This inequality means that set $K$ is not ``good'' for {\em some} $g$.

Thus, probability of set $K$ not being ``good'' {\em at least for some} non-unit $g$ does not exceed
$(|G|-1) / |G|$.

So, ``good'' automorphism set exists with probability of $1/|G|$.
\end{proof}

\section{On dimension of subspaces and group in quantum hash function}
\label{sec:restrictions}

\begin{defn}\label{defn:reflection}
  A complex reflection is non-trivial element that fix a complex hyperplane in space pointwise. A
  $p$--fold reflection matrix has characteristic roots 1 (repeated $n-1$ times) and $\theta$, a
  primitive $p$--th root of unity (cf. \cite{ShephardTodd53}).
\end{defn}

\begin{defn}\label{defn:uggr}
  Unitary group generated by reflections (u.g.g.r.) --- any finite group acting on a
  finite-dimensional complex vector space that is generated by complex reflections.
\end{defn}

The unitary groups generated by reflections were classified by Shephard and Todd in~\cite{ShephardTodd54}. 

Let us have quantum hash function:
\[
\ket{\Psi_{h,G,K,f,m}(x)}
\]

Because of Shephard and Todd result group $G$ and dimension $m$ cannot be arbitrary: group must be
subgroup of product of u.g.g.r., and dimension cannot be less than minimal product of dimensions of
corresponding u.g.g.r.

In other words, holds
\begin{thm}\label{thm:hash-uggr}
  In arbitrary quantum hash function $\ket{\Psi_{h,G,K,f,m}(x)}$ group $G$ is subgroup of some
  product of u.g.g.r. in $\HH{m}$; and vice versa, dimension $m$ cannot be less than dimension of
  space in which such product of u.g.g.r. can exist.
\end{thm}

E.g., let we want to construct quantum hash function in which classical hash function maps input to
some permutation on $k$ points. Then it is impossible to find homomorphism into space with
dimension less than $k$.

\begin{proof}
  Let us prove that group $G$ must be subgroup of product of u.g.g.r.

  Let group $G$ consists of $g_1$, $g_2$, \ldots, $g_t$. Because group is finite, there is some $r$
  that is group order:
  \[
  g_i^r = e
  \]

  Because $f(g_i) = U_i$ is a unitary matrix, its characteristic roots are some $r$--th roots of unity:
  \[
  f(g_i)^r = f(g_i^r) = f(e) = E
  \]

  Thus each of $U_i$ can be represented in form:
  \[
  U_i = V_i^T A_{i0} A_{i1} A_{is_i} V_i,
  \]
  where $A_{ij}$ is some reflection matrix, and $V_i$ is some unitary matrix.

  All elements of form $V_i^T A_{ij} V_i$ are some reflections in unitary space and are elements of
  some u.g.g.r.
\end{proof}





\end{document}